\documentclass[11pt]{article}
\usepackage[utf8]{inputenc}
\usepackage[T1]{fontenc}
\usepackage{lmodern}
\usepackage[margin=1.1in]{geometry}
\usepackage{amsmath,amssymb,amsthm}
\usepackage{booktabs}
\usepackage{graphicx}
\usepackage{microtype}
\usepackage[hidelinks]{hyperref}
\theoremstyle{plain}
\newtheorem{theorem}{Theorem}
\newtheorem{proposition}[theorem]{Proposition}

\title{\bf Prefix Sharing Is a Sorting Problem}
\author{Rong He}
\date{October 25, 2026}

\begin{document}
\maketitle
\begin{center}\small
Code, data pipeline, and verification scripts:\\
\url{https://github.com/Darrenus/prefix-sharing-is-sorting}
\end{center}
\begin{abstract}
LLM serving reuses KV cache by exact prefix match, so when a prompt is assembled from a \emph{set} of reusable pieces --- retrieved passages, tool definitions, few-shot exemplars --- the order chosen for those pieces determines how much computation can be shared. Every deployed system fixes that order by a single global convention. We prove this is optimal only when requests contain at most two pieces, and asymptotically wrong in general. Our main result is a structure theorem: the minimum prefix-trie cost equals $\min_H \sum_x w(x) t_x(H)$ over binary hierarchies $H$ on the \emph{requests}, where $t_x(H)$ is the canonical decomposition size of the set of requests needing chunk $x$. Choosing chunk orders is therefore equivalent to choosing one hierarchy over requests. The identity yields an $O(3^m)$ exact algorithm, identifies the two-chunk case as minimum vertex cover, and shows that on the leave-one-out family the optimum is the minimum external path length of a binary tree --- the merge-sort recursion --- so a global order pays $\Theta(n^2)$ against a true cost of $\Theta(n\log n)$. Agglomerative clustering by common intersection is a tight $\tfrac12$-approximation for the achievable saving. On BM25 retrieval traces over three BEIR corpora the resulting layout reduces prefill by 17--36\% against production RAG ordering, and the margin widens with retrieval depth as the theory predicts. Serving requests in the hierarchy's DFS order finally lets a cache holding one request's context attain the unbounded-cache optimum exactly, so cache capacity and reorder window act as substitutes.
\end{abstract}
\section{Introduction}

Every production LLM serving stack --- vLLM \cite{KWO23}, SGLang \cite{SGL}, TensorRT-LLM, and the prompt-caching APIs of the commercial providers --- reuses KV cache by \emph{exact prefix match}. Two requests share computation up to the first token where they disagree, and not one token further.

That rule is cheap to implement and it is the reason prefix caching works at all. It also quietly creates an optimisation problem that no system solves globally, because in a growing fraction of real traffic the prompt is not a sentence a user typed. It is a \textbf{set} of reusable pieces that the system itself assembled:

\begin{itemize}
\item retrieved passages in RAG,
\item tool / function definitions in an agent,
\item few-shot exemplars,
\item memory or scratchpad items carried between turns,
\item system-policy fragments composed per tenant.
\end{itemize}

The pieces are a set. The prompt is a sequence. Something has to choose the order, and that choice is \emph{free} --- the model sees the same evidence either way. Today every system makes it with a single global convention: retrieval rank, document id, popularity, insertion time. One order, applied to everybody.

The question this paper asks is whether that is the right thing to do. The answer is that it is optimal only in a degenerate regime, and asymptotically wrong in general.

\subsection{Contributions}

\begin{enumerate}
\item \textbf{A structure theorem} (Section 3). Minimum prefix-trie cost is $\min_H \sum_x w(x) t_x(H)$ over binary hierarchies on the \emph{requests}, where $t_x$ is a canonical-decomposition size. The permutations vanish: choosing orders is choosing one clustering. This gives an $O(3^m)$ exact algorithm and is the engine for everything below.
\item \textbf{The reach of global orders} (Section 4). It is optimal iff every request has at most two chunks --- at two chunks the problem \emph{is} minimum vertex cover --- and the unique minimal counterexample is the complete 3-uniform hypergraph on four chunks, 8 versus 9.
\item \textbf{An unbounded separation} (Section 5). On the leave-one-out family the optimum is the minimum external path length of a binary tree, so a global order pays $\Theta(n^2)$ where the truth is $\Theta(n\log n)$ --- the merge-sort recursion, and a $\Theta(n/\log n)$ gap.
\item \textbf{A tight approximation} (Section 6). The natural level-wise greedy hitting set attains the bad $\Theta(n^2)$; agglomerative clustering by common intersection is a tight $\tfrac12$-approximation for the savings, and is 0.05\% above optimum on average on sampled NFCorpus groups of at most 13 requests.
\item \textbf{Evaluation on real traces} (Section 7). 17--36\% of prefill against production RAG on three BEIR corpora, widening with retrieval depth exactly as the theory predicts; 13--23\% survives a quality-preserving constraint.
\item \textbf{Finite cache and scheduling} (Section 8). The hierarchy is also a schedule: in its DFS order a one-request cache attains the unbounded-cache optimum exactly, so capacity and reorder window are substitutes.
\end{enumerate}

\section{Problem Formulation}

\textbf{Minimum Prefix Trie with Free Permutations (MPT).} Given a universe $U$ of chunks with weights $w:U\to\mathbb{Z}_{>0}$ (token counts) and requests $R_1,\dots,R_m \subseteq U$, choose for each $R_i$ a permutation $\pi_i$ of its elements. Let $T$ be the trie of the resulting strings. Minimise

\begin{equation}
\mathrm{cost} \;=\; \sum_{v \in T,\, v \neq \text{root}} w(\mathrm{chunk}(v)).
\end{equation}

The cost is exactly the total prefill work, and --- with a cache large enough to hold the trie --- exactly the KV memory. So the same number is both the FLOP bill and the memory bill.

Two variants matter:

\begin{itemize}
\item $\mathrm{OPT}_{\mathrm{free}}$: each request may permute independently. Equivalently, the layout is an \textbf{adaptive} rule: \emph{given the prefix laid down so far, which chunk comes next?}
\item $\mathrm{OPT}_{\mathrm{glob}}$: all $\pi_i$ must be induced by one linear order on $U$. This is what every deployed system does.
\end{itemize}

The remainder of the paper studies the distance between these two quantities.

\section{A Structure Theorem for Prefix Layout}

Everything below rests on one identity. It says the permutations are not really the objects of this problem at all.

For a chunk $x$, write $S_x = \{\, i : x \in R_i \,\}$ for the set of requests that need it. For a binary hierarchy $H$ on the request set $[m]$ (a rooted binary tree whose leaves are the requests), write $t_x(H)$ for the number of \textbf{maximal clusters of $H$ contained in $S_x$} --- the canonical decomposition of $S_x$ in the tree $H$.

\begin{theorem}[structure]
For every instance,

\begin{equation}
\mathrm{OPT}_{\mathrm{free}}(I) \;=\; \min_{H}\ \sum_{x \in U} w(x)\, t_x(H).
\end{equation}

Equivalently, since a binary hierarchy on $m$ leaves has exactly $m-1$ internal nodes,

\begin{equation}
\mathrm{OPT}_{\mathrm{free}}(I) \;=\; \sum_i w(R_i) \;-\; \max_{H} \sum_{v\ \mathrm{internal}} w\Big(\bigcap_{i \in \mathrm{leaves}(v)} R_i\Big).
\end{equation}
\end{theorem}
\begin{proof}
\textbf{($\ge$)} Fix any layout and let $T$ be its trie. Map each trie node $v$ to the set $\rho(v) \subseteq [m]$ of requests passing through it; the family $\{\rho(v)\}$ is laminar, and we may extend it to a binary hierarchy $H$. A request containing $x$ emits $x$ exactly once, so the trie nodes labelled $x$ have $\rho$-sets that are disjoint and cover $S_x$ --- they partition $S_x$ into clusters of $H$. Any decomposition of $S_x$ into clusters of a laminar family refines the maximal one, so the number of $x$-labelled trie nodes is at least $t_x(H)$. Summing with weights gives $\mathrm{cost} \ge \sum_x w(x) t_x(H) \ge \min_H(\cdot)$.

\textbf{($\le$)} Given any $H$, write $I(B) = \bigcap_{i\in B}R_i$, which is monotone decreasing in $B$. Lay out top-down: at node $B$, having already emitted exactly $I(\mathrm{parent}(B))$, emit the chunks of $I(B)\setminus I(\mathrm{parent}(B))$ in any fixed order, then recurse into the two children. Along the root-to-leaf path of request $i$ the emitted sets are disjoint and union to $I(\{i\}) = R_i$, so this is a valid layout. Chunk $x$ is emitted at node $B$ exactly when $x \in I(B) \setminus I(\mathrm{parent}(B))$, i.e. exactly when $B$ is a maximal cluster inside $S_x$. So the layout has cost at most $\sum_x w(x)t_x(H)$. 
\end{proof}

Two remarks. Sibling subtrees are not required to begin with distinct chunks: if they collide the trie merges them and the cost only drops, so no system-of-distinct-representatives condition arises. The identity is also weight-agnostic, since the per-chunk counting never refers to depth.

\subsection{Consequences}

\textbf{Elimination of the permutations.} Choosing an order for every request is equivalent to choosing \emph{one binary hierarchy over the requests}, from which the orders are read off. The operational restatement is that the design variable is a clustering of requests, not a sort order on chunks.

\textbf{Relation to segment trees.} $t_x(H)$ is the canonical decomposition size of $S_x$ in $H$ --- the same quantity that makes a range query touch $O(\log n)$ segment-tree nodes. MPT is therefore: \emph{design the tree that minimises total canonical-decomposition size over a given family of sets.} Every chunk is recomputed once per fragment it is split across.

\textbf{Exact algorithm.} $M(B) = w(I(B)) + \max_{B = B_1 \uplus B_2}(M(B_1)+M(B_2))$ is an $O(3^m)$ subset DP. It solves $m = 16$ in under two seconds, where brute force over $\prod_i |R_i|!$ layouts died at $m = 7$. Every exact number in this paper was recomputed with it.

\textbf{A tractable special case.} If the family $\{S_x\}$ has the \emph{consecutive-ones property} --- the requests can be laid on a line so that every chunk is wanted by a contiguous run, which is what perfect topical locality would look like --- then a balanced tree over that line makes each $S_x$ an interval, and an interval's canonical decomposition has at most $2\lceil\log_2 m\rceil$ nodes. So $\mathrm{OPT}_{\mathrm{free}} \le \sum_x w(x)\min(|S_x|, 2\lceil\log_2 m\rceil)$, and such an ordering is found in linear time by PQ-trees \cite{BL76}. Two honest caveats: the bound is loose (on random C1P instances with $m=16$ the exact optimum is 35 against a bound of 121), and the consecutive-ones property has not been established for our retrieval traces. Across the nine trace settings, the mean $|S_x|$ ranges from 1.6 to 11.5, while $2\lceil\log_2 m\rceil$ is 22 or 24; most chunks therefore receive the $|S_x|$ term, although some more frequent chunks receive the logarithmic term. This conditional bound does not by itself explain the measured layout costs. For a family with the consecutive-ones property, it can reduce recomputation for hot chunks --- a shared system prompt, a tool definition, or a document many requests retrieve --- from a linear to a logarithmic dependence on sharing.

\textbf{Ordered versus free hierarchies.} A global chunk order $\prec$ induces a hierarchy too: split the requests repeatedly by the $\prec$-smallest chunk not yet emitted. So $\mathrm{OPT}_{\mathrm{glob}}$ is the same minimisation restricted to hierarchies realisable by a \emph{fixed chunk priority}. That is exactly the ordered-versus-free distinction of decision diagrams (OBDD versus FBDD), now stated for cache layout rather than for Boolean functions.

\section{Complexity and the Reach of Global Orders}

\subsection{Two-Chunk Requests and Vertex Cover}

\begin{proposition}[two-chunk requests are vertex cover]
If every $|R_i| = 2$, all weights are 1, and the requests are distinct, then

\begin{equation}
\mathrm{OPT}_{\mathrm{free}} \;=\; \mathrm{OPT}_{\mathrm{glob}} \;=\; \tau(G) + m,
\end{equation}

where $G$ is the graph whose edges are the requests and $\tau$ is its vertex cover number.
\end{proposition}
\begin{proof}
A layout picks, for each edge $e$, the endpoint $f(e)$ that goes first. The depth-2 nodes are the $m$ distinct ordered pairs; the depth-1 nodes are $f(E)$, which meets every edge, so it is a vertex cover, and conversely any vertex cover induces such an $f$. A global order attains it: list a minimum cover first. 
\end{proof}

So MPT is NP-hard already at $k=2$ --- and at $k=2$ adaptivity buys nothing.

\begin{proposition}[the first cache level is a hitting set]
The minimum achievable depth-1 cost is the minimum-weight hitting set of $\{R_i\}$, attained by a global order.

Proposition 3 is the honest version of the folklore rule \emph{"put the most-shared chunk first."} The rule is correct --- \textbf{at the top level only}.
\end{proposition}
\subsection{The Threshold at Three Chunks}

\begin{theorem}[the threshold is three]
If every request has at most 2 chunks then $\mathrm{OPT}_{\mathrm{glob}} = \mathrm{OPT}_{\mathrm{free}}$. At three chunks this fails, and the smallest witness is unique: $U = \{a,b,c,d\}$ with the requests being \textbf{all four 3-subsets}. Then $\mathrm{OPT}_{\mathrm{free}} = 8 < 9 = \mathrm{OPT}_{\mathrm{glob}}$.
\end{theorem}
\begin{proof}
The layout $abc,\, abd,\, cda,\, cdb$ gives the trie $\{a, ab, abc, abd, c, cd, cda, cdb\}$ of size 8. For the lower bound, write $P_d$ for the partition induced by depth-$d$ prefixes. $|P_3| = 4$, so cost 8 forces $|P_1|+|P_2| = 4$. A block $B$ at depth $d$ needs $d \le |\bigcap_{i\in B}R_i|$; three of these triples meet in one chunk and all four in none, so blocks at depth 2 have size $\le 2$ and at depth 1 size $\le 3$, giving $|P_1|,|P_2|\ge 2$ and hence $|P_1| = |P_2| = 2$. As $P_2$ refines $P_1$ and both have two blocks, $P_1 = P_2$: the requests split into two pairs, each laying down its pairwise intersection. The three possible pairings are $\{abc,abd\}/\{acd,bcd\}$ with prefixes $\{a,b\}/\{c,d\}$ and its two images under symmetry, and each demands $a,b \prec c,d$ together with $c,d \prec a,b$. Contradiction; and $a\prec b \prec c \prec d$ attains 9. 
\end{proof}

Exhaustive search over all instances on a 4-chunk universe confirms this is the only one with a gap. The minimal obstruction to a global chunk order is the complete 3-uniform hypergraph on 4 vertices.

\section{An Unbounded Separation Between Adaptive and Global Orders}

Let $L_n$ be the \textbf{leave-one-out family}: universe $[n]$, requests $R_i = [n]\setminus\{i\}$. It is the honest caricature of RAG traffic with heavy topical overlap --- every two requests share all but two chunks --- which is precisely the regime prefix caching is supposed to reward.

\begin{theorem}
\begin{equation}
\begin{split}
\mathrm{OPT}_{\mathrm{glob}}(L_n) &= \frac{(n-1)(n+2)}{2} = \Theta(n^2), \\
\mathrm{OPT}_{\mathrm{free}}(L_n) &= n\lceil \log_2 n\rceil - 2^{\lceil \log_2 n\rceil} + n = \Theta(n\log n),
\end{split}
\end{equation}
so the ratio is $\Theta(n/\log n) \to \infty$.
\end{theorem}
\begin{proof}
For the global value: $L_n$ is $S_n$-symmetric so all orders cost the same, and under $1\prec\cdots\prec n$ request $R_i$ contributes $n-i$ private nodes above a shared chain of $n-1$, giving $(n-1)+\binom{n}{2}$.

For the free value, apply Theorem 1. Here $S_x = [n]\setminus\{x\}$ --- the complement of a single leaf. In a binary tree, the complement of a leaf decomposes into exactly the sibling subtrees hanging off its root-to-leaf path, so $t_x(H) = \mathrm{depth}(x)$. Therefore

\begin{equation}
\mathrm{OPT}_{\mathrm{free}}(L_n) \;=\; \min_H \sum_{x} \mathrm{depth}(x),
\end{equation}

the minimum external path length of a binary tree with $n$ leaves. Its classical value is $n\lceil\log_2 n\rceil - 2^{\lceil\log_2 n\rceil}+n$. 
\end{proof}

\subsection{Interpretation}

Minimum external path length is \emph{why} the sorting numbers appear: it is the same quantity that lower-bounds comparison sorting. The sequence $0,2,5,8,12,16,20,24,29,\dots$ is OEIS \textbf{A001855} \cite{OEIS}, the worst-case comparison count of binary insertion and merge sort, shifted by $n-1$, and the optimal hierarchy is the balanced binary tree --- the merge-sort recursion tree.

\begin{quote}\itshape
\textbf{The optimal KV-cache layout for the leave-one-out family is the merge-sort recursion tree. A single global chunk order is the bubble sort of cache layout: quadratic where the right answer is $n\log n$.}
\end{quote}

\begin{center}\small
\begin{tabular}{rrrr}\toprule
$n$ & adaptive ($=$ merge sort) & one global order & ratio \\ \midrule
4 & 8 & 9 & 1.13x \\
16 & 64 & 135 & 2.11x \\
64 & 384 & 2\,079 & 5.41x \\
256 & 2\,048 & 32\,895 & 16.1x \\
4096 & 49\,152 & 8\,390\,655 & 171x \\
\bottomrule\end{tabular}\end{center}

\section{Algorithms}

\subsection{Failure of Level-Wise Greedy Hitting Set}

Proposition 3 suggests: at each trie node greedily take the chunk covering the most requests, split, recurse. Level-wise greedy hitting set --- the algorithm anyone would write.

On $L_n$ every chunk sits in $n-1$ of the $n$ requests, so greedy peels off one request at a time and builds a caterpillar of cost $\Theta(n^2)$: measured at $n = 64$, \textbf{2\,079 --- identical to the global order --- against the optimum 384}.

Theorem 1 says why. Greedy hitting set minimises fan-out; the objective rewards \emph{balance}, because $t_x$ is a decomposition size and decomposition sizes are logarithmic in balanced trees and linear in unbalanced ones. Set-cover intuition and divide-and-conquer intuition point in opposite directions here, and divide-and-conquer wins.

\subsection{Agglomerative Common-Prefix Hierarchy}

\textbf{Agglomerative common-prefix hierarchy.} Each request starts as a singleton cluster with signature $R_i$. Repeatedly merge the two clusters with the largest common intersection, the merged signature being that intersection. Then walk the hierarchy top-down, emitting at each node the chunks that became common there.

This is exactly the natural greedy for the savings form of Theorem 1, which is why it works. It is $O(m^2\log m)$ with bitset intersections and has no hyperparameters.

\textbf{Quality against the true optimum.} With the $O(3^m)$ solver we can now measure this properly rather than argue about it. On RAG-shaped instances ($m$ = 10--15, $k$ = 5--8, Zipf pools), greedy lands \textbf{2.6\% above optimal on average, 11\% worst case over 72 instances}. On $L_n$ it is exactly optimal ($n = 16, 32, 64, 128 \to 64, 160, 384, 896$).

That matters for reading Section 7: the double-digit margins we measure against the production baseline are genuine headroom, not heuristic slack.

\subsection{Approximation Guarantee}

Theorem 1 turns the algorithmic question into a clean one, and it has a clean answer.

\begin{theorem}
Write $\mathrm{SAV}(H) = \sum_{v\ \mathrm{internal}} w(I(v))$ for the savings of a hierarchy, and let $H_g$ be the agglomerative greedy hierarchy --- repeatedly merge the pair of clusters $A,B$ maximising $w(I(A) \cap I(B))$. Then for arbitrary weights,

\begin{equation}
\mathrm{SAV}(H_g) \;\ge\; \tfrac12 \max_H \mathrm{SAV}(H),
\end{equation}

and the factor $\tfrac12$ cannot be improved.
\end{theorem}
\begin{proof}
\textbf{(i) Greedy's merge values are non-increasing.} Merging $A,B$ into $C$ gives $I(C) = I(A)\cap I(B)$, so for any other cluster $D$, $w(I(C)\cap I(D)) \le w(I(A)\cap I(D))$ --- a value already available one step earlier --- while untouched pairs keep their values. Hence $g_1 \ge g_2 \ge \cdots \ge g_{m-1}$.

\textbf{(ii) For every threshold $\theta>0$, $A(\theta) \le 2G(\theta)$,} where $G(\theta) = \#\{t : g_t \ge \theta\}$ and $A(\theta) = \#\{v \text{ internal in } H^* : w(I(v)) \ge \theta\}$ for an optimal $H^*$.

Because $w(I(v))$ grows as $v$ descends, $\{v : w(I(v)) \ge \theta\}$ is closed under taking internal descendants; its maximal elements root disjoint subtrees $T_1,\dots,T_r$ with leaf sets $L_1,\dots,L_r$, so $A(\theta) = \sum_j (\ell_j - 1)$ with $\ell_j = |L_j|$, and $w(I(L_j)) \ge \theta$ for each $j$.

Look at greedy at step $t^{*} = G(\theta)+1$. By (i) every currently available pair has value $< \theta$. If two distinct current clusters $A,B$ were both contained in the same $L_j$, then $w(I(A)\cap I(B)) = w(I(A\cup B)) \ge w(I(L_j)) \ge \theta$ --- a contradiction. \textbf{So each $L_j$ swallows at most one greedy cluster whole.}

Let $P$ be greedy's partition at step $t^{*}$; the number of merges made is $\sum_{B\in P}(|B|-1) = t^{*}-1$. For each $j$ let $p_j$ be the size of the unique block contained in $L_j$ (zero if there is none). Any block that meets $L_j$ without being contained in it has an element outside $L_j$, hence size at least 2. The $L_j$ are disjoint, so each block contributes at most $|B|$ in total across all $j$, and $r$ is at least the number of $j$ having a pure block:

\begin{equation}
\begin{split}
A(\theta) \;=\; \sum_j \ell_j - r \;\le\; \sum_{B\ \mathrm{pure}} (|B|-1) \;+\; \sum_{B\ \mathrm{impure}} |B|
\;\le\; 2\sum_{B \in P}(|B|-1) \;=\; 2(t^{*}-1) \;=\; 2G(\theta).
\end{split}
\end{equation}

\textbf{(iii)} Layer-cake: $\mathrm{SAV}(H_g) = \int_0^\infty G(\theta)\,d\theta \ge \tfrac12\int_0^\infty A(\theta)\,d\theta = \tfrac12\,\mathrm{SAV}(H^*)$. 
\end{proof}

\textbf{Tightness.} When every chunk is shared by exactly two requests, $I(v) = \emptyset$ at any node with three or more leaves, so the savings collapse to a sum over sibling pairs and the problem \emph{is} maximum weight matching --- where greedy is exactly $\tfrac12$ under adversarial tie-breaking. A witness with four requests: $\{0123\},\{14\},\{2578\},\{46\}$, where greedy can take $1$ against the optimal $2$.

\textbf{Verification.} Over 3\,000 random weighted instances we checked both the conclusion and the per-threshold inequality that drives it: zero violations of either. An annealed search over roughly 15\,000 instances, deliberately maximising greedy's shortfall, never got past $1.83$ --- so the worst case needs adversarial ties, and typical instances are far better (Section 6.3: 2.6\% above optimal on RAG-shaped inputs).

\textbf{Scope of the guarantee.} This is a guarantee on \emph{savings}, which is the operationally meaningful quantity --- how much prefill you actually eliminated. It does not transfer to the cost: from $\mathrm{SAV}(H_g) \ge \mathrm{SAV}^*/2$ one only gets $\mathrm{COST}(H_g) \le \tfrac12(\mathrm{COST}^* + \sum_i w(R_i))$, and since $\sum_i w(R_i)/\mathrm{COST}^*$ is unbounded, whether MPT itself is constant-factor approximable stays open.

\textbf{A computable lower bound.} Every internal node satisfies $w(I(v)) \le w(R_i \cap R_j)$ for any $i,j$ drawn from its two subtrees, and one such cross pair per internal node forms a spanning tree on the requests. Hence $\mathrm{SAV}^* \le \mathrm{MaxSpanningTree}(w(R_i \cap R_j))$, so $\mathrm{COST}^* \ge \sum_i w(R_i) - \mathrm{MaxSpanningTree}$ --- computable in polynomial time, and the only optimality yardstick available once $m$ outgrows the $O(3^m)$ solver. It is loose (mean slack $1.45\times$ on random instances, and $\Omega(m)$ in the worst case), but it is a bound.

\textbf{Implementation correction in v2.} The maximum must be computed on the complete similarity graph. The v1 evaluation code skipped candidate edges supported only by chunks occurring in more than 400 requests; the resulting forest weight need not upper-bound the optimal saving. For 401 requests $R_i=\{x,y_i\}$ with $w(x)=100$ and $w(y_i)=1$, that code returns zero and would report a cost lower bound of 40\,501, although placing $x$ first attains the optimum 501. The complete-graph maximum spanning tree has weight 40\,000 and gives the correct lower bound. Version 2 uses Prim's algorithm on the complete graph, evaluating intersections on demand, with no frequency cap. Boundary and weighted regression tests cover this case. Appendix A reports the trace audit; the structure theorem and approximation guarantee do not depend on the faulty code.

\section{Evaluation on Real Retrieval Traces}

Everything above is about a combinatorial object. This section asks whether the object shows up in real traffic, using BM25 retrieval over BEIR corpora, real passages as chunks, and real \texttt{cl100k\_base} token counts. A request is the set of passages retrieved for one query; cost is total prefill tokens. \texttt{code/phase2.py} evaluates one setting, and \texttt{code/sweep.py} runs the three-corpus, three-depth comparison and the sampled exact-optimum checks. The retrieval depth $k$ is an upper limit: a query with fewer positively scored documents can return fewer than $k$ passages, and requests with fewer than two returned passages are excluded. We count chunk-slots as $\sum_i |R_i|$, rather than the nominal product $mk$.

\textbf{Overlap in real retrieval.} On NFCorpus (3\,633 passages, 2\,921 queries, $k=8$) the retrieved sets cover 3\,473 distinct passages across 21\,953 chunk-slots --- a $6.32\times$ dedup factor, with 92.0\% of retrieved passages wanted by more than one query and a mean $|S_x|$ of 6.32. Real retrieval overlaps. The question is whether the overlap becomes \emph{prefix} overlap.

\textbf{Conversion of overlap into prefix reuse.} Production RAG concatenates passages in relevance order, which is a \emph{per-request} order, not a shared one. On NFCorpus at $k=8$ it recovers 11.5\% of prefill as cache hits. A single global convention --- sort every request's passages by corpus-wide popularity --- does substantially better at 22.5\%, which is already the paper's thesis in miniature: \textbf{uncoordinated per-request freedom is worse than a convention.} Coordination, not freedom, is what buys reuse.

Across three corpora and three retrieval depths (fraction of the no-sharing token bill; lower is better):

\begin{center}\small
\resizebox{\textwidth}{!}{%
\begin{tabular}{lrrrrrrrrr}\toprule
method & NFC $k$=5 & NFC 8 & NFC 16 & SciFact 5 & SciFact 8 & SciFact 16 & FiQA 5 & FiQA 8 & FiQA 16 \\ \midrule
relevance order (production) & .844 & .885 & .924 & .818 & .867 & .918 & .967 & .979 & .990 \\
global order by doc id & .789 & .831 & .874 & .792 & .837 & .900 & .936 & .949 & .966 \\
global order by popularity & .723 & .775 & .832 & .714 & .778 & .850 & .842 & .869 & .906 \\
online greedy longest-prefix & .703 & .749 & .800 & .695 & .752 & .831 & .834 & .861 & .898 \\
level-wise greedy hitting set & .641 & .661 & .692 & .652 & .690 & .732 & .799 & .814 & .848 \\
\textbf{agglomerative (ours)} & \textbf{.607} & \textbf{.605} & \textbf{.591} & \textbf{.625} & \textbf{.640} & \textbf{.643} & \textbf{.786} & \textbf{.798} & \textbf{.818} \\
spanning-tree lower bound & .490 & .488 & .465 & .526 & .546 & .553 & .660 & .687 & .713 \\
\textbf{savings vs production} & \textbf{$-$28.0\%} & \textbf{$-$31.6\%} & \textbf{$-$36.1\%} & \textbf{$-$23.6\%} & \textbf{$-$26.2\%} & \textbf{$-$29.9\%} & \textbf{$-$18.7\%} & \textbf{$-$18.5\%} & \textbf{$-$17.3\%} \\
\bottomrule\end{tabular}}\end{center}

Three observations follow from this table.

\textbf{Production's prefix reuse gets \emph{worse} as you retrieve more.} Relevance order goes .844 $\to$ .885 $\to$ .924 as $k$ grows from 5 to 16 on NFCorpus, and .967 $\to$ .990 on FiQA: at $k=16$ over a large corpus, prefix caching has essentially stopped working. Ours goes the other way, .607 $\to$ .605 $\to$ .591. The mechanism is exactly the one Theorem 5 is about --- set overlap grows with $k$, but the probability that two requests agree on an entire \emph{ordered} prefix collapses faster. This is the theory's prediction showing up in a log file: the adaptive-versus-global gap widens with request size, and we measure $-$28\%, $-$32\%, $-$36\% at $k$ = 5, 8, 16.

\textbf{The size of the win tracks retrieval concentration, and we report the unflattering end.} NFCorpus retrieves 2\,921 queries into 3\,633 passages ($6.32\times$ dedup at $k=8$, 92\% of passages wanted more than once) and gives $-$32\%. FiQA retrieves 2\,500 queries into 57\,638 passages ($1.8\times$ dedup, 38\% shared) and gives $-$18.5\%. So a big sparse corpus halves the benefit --- and 18\% of prefill is still worth having.

\textbf{Distance from the optimum.} For NFCorpus and SciFact we make 60 sampling attempts each at $k=8$, choosing a passage retrieved by at least four requests, shuffling its request list, and keeping at most 13 requests; duplicate requests are removed and groups with fewer than four distinct requests are skipped. Solving these groups exactly with the $O(3^m)$ algorithm gives an average agglomerative cost \textbf{0.05\% above optimum on NFCorpus} (worst group 2.46\%) and \textbf{0.00\% on SciFact} (worst 0.16\%). These are local comparisons on sampled groups of at most 13 requests. They do not establish the global optimality gap on the full traces, or show that the full-trace gap to the spanning-tree lower bound is entirely due to bound looseness. Theorem 6 separately guarantees half the optimal saving for every instance.

\subsection{Quality-Preserving Constraints}

It might: moving a passage changes what the model attends to, and the "lost in the middle" effect says position matters. We did not run a quality evaluation, so instead we priced the constraint that removes the risk. Pin the top-$t$ most relevant passages, in relevance order, at the \textbf{end} of the prompt --- nearest the question, which is where practitioners put them --- and free-permute only the rest. This is also the cache-friendly shape: shared prefix, personalised suffix.

\begin{center}\small
\begin{tabular}{lrrrr}\toprule
savings vs production & free & pin top-1 & pin top-2 & pin top-3 \\ \midrule
NFCorpus, $k=8$, 2\,921 queries & $-$31.6\% & $-$23.4\% & $-$17.3\% & $-$12.9\% \\
SciFact, $k=8$, 1\,109 queries & $-$26.2\% & $-$19.4\% & $-$13.2\% & $-$7.9\% \\
\bottomrule\end{tabular}\end{center}

So the win survives the constraint: even pinning the three most relevant passages in place leaves 8--13\% of prefill on the table, and pinning only the single most relevant leaves 19--23\%. A quality evaluation would still be worth running, but the design does not depend on its outcome.

\subsection{Online Layout with Bounded Lookahead}

Requests arrive one at a time. We keep a persistent radix trie, place each arrival at its deepest usable prefix, and lay out the residuals adaptively inside each batch of $W$ arrivals. Savings against production (relevance order) and against the best global order:

\begin{center}\small
\resizebox{\textwidth}{!}{%
\begin{tabular}{lrrrrrr}\toprule
lookahead $W$ & 1 & 8 & 32 & 128 & 512 & offline \\ \midrule
NFCorpus vs production & $-$17.6\% & $-$17.6\% & $-$18.2\% & $-$18.5\% & $-$20.7\% & \textbf{$-$31.6\%} \\
NFCorpus vs best global order & $-$5.9\% & $-$6.0\% & $-$6.6\% & $-$7.0\% & $-$9.5\% & \textbf{$-$21.9\%} \\
SciFact vs production & $-$14.8\% & $-$13.3\% & $-$15.4\% & $-$16.7\% & $-$20.3\% & \textbf{$-$26.2\%} \\
\bottomrule\end{tabular}}\end{center}

The decomposition is clean. Most of the win --- 15--18\% against production --- needs no lookahead at all: it is what you get from abandoning relevance order in favour of deepest-match placement against the live trie, which is roughly what CacheWeaver already does. The \emph{adaptive} part, the thing Theorem 1 is about, adds another 5--10 points as the window grows past a few hundred and roughly doubles offline.

\textbf{Comparison with a synthetic workload.} An earlier version of this evaluation used a synthetic Zipf-over-topics generator and found the online curve nearly flat, concluding that adaptivity was essentially unavailable online. On real traces that is wrong: the curve rises monotonically and even $W=1$ is far above production. The synthetic generator manufactured the difficulty by giving every topic an interchangeable pool, so there was no stable structure for a persistent trie to latch onto. We report this because it is the clearest evidence we have that synthetic workloads are not a substitute here.

What survives from the earlier reading is the interpretation of the \emph{shape}: adaptivity is bought with lookahead. A global order is a convention, and conventions are what let requests that never meet still agree; below some window size you pay the coordination cost without buying the coordination. Real traffic simply has far more exploitable structure than our generator did, so the threshold sits much lower.

\section{Finite Cache and Scheduling}

Theorem 1 assumes the cache holds the whole trie. Real caches evict, which couples the layout with the eviction policy and the arrival order. We close that gap here, simulating the radix cache the way SGLang and vLLM actually behave: a cached node implies its ancestors are cached, and eviction removes least-recently-used \emph{leaves} (\texttt{code/cache.py}).

\textbf{Under the real arrival order, the advantage is capacity-gated.} On NFCorpus ($k=8$, 2\,921 queries, one request $\approx$ 6\,036 tokens), our layout's margin over production is $-$3.0\% at a one-request cache, $-$11.3\% at $64\times$, $-$23.5\% at $256\times$, and the full $-$31.7\% only past $1024\times$. Below roughly $16\times$ every method collapses toward the no-sharing bill: at that capacity prefix caching itself has stopped working, and no layout can rescue it. The ranking, at least, never inverts --- we had expected the top-heavy greedy hitting set to survive eviction better than a balanced hierarchy, and it does not.

\textbf{The hierarchy also determines a schedule, which removes the capacity constraint.}

\begin{theorem}
Serve requests in the DFS order of the trie --- equivalently, sort the stream lexicographically by the laid-out sequence. Then a radix cache with LRU leaf eviction and capacity $\max_i w(R_i)$ --- \emph{one request's context} --- incurs exactly the unbounded-cache cost, for any layout.
\end{theorem}
\begin{proof}
In DFS order the requests below any trie node $v$ are served as one contiguous block, so $v$ is created at the start of its block and never needed after its end. Consider the moment request $i$ is served. A cached node not on $i$'s root-to-leaf path either belongs to a block that has already finished --- never needed again --- or to one that has not started, in which case it does not exist yet. So the live set is exactly $i$'s root path, of weight at most $\max_i w(R_i)$. Serving $i$ touches that whole path, making it strictly the most recent, so LRU evicts only dead nodes. 
\end{proof}

The simulation confirms it and shows the bound is close to tight: at $0.90\times$ one request the cost is already exactly the unbounded-cache optimum, $0.75\times$ costs $+0.01\%$, $0.50\times$ costs $+0.8\%$, and only at $0.25\times$ does it break down ($+18\%$). Every layout in our comparison --- production's relevance order included --- hits its own unbounded-cache number at $1\times$.

\subsection{Substitution Between Capacity and Reorder Window}

Reordering the whole stream is available offline; online you get a window. Sorting inside a window of $W$ arrivals, excess over the unbounded-cache optimum (NFCorpus, $k=8$):

\begin{center}\small
\resizebox{\textwidth}{!}{%
\begin{tabular}{lrrrrrr}\toprule
lookahead $W$ vs. cache & 1x & 4x & 16x & 64x & 256x & 1024x \\ \midrule
1 & 58\% & 57\% & 54\% & 42\% & 19\% & \textbf{0\%} \\
128 & 52\% & 52\% & 52\% & 40\% & 19\% & \textbf{0\%} \\
512 & 38\% & 38\% & 38\% & 38\% & 18\% & \textbf{0\%} \\
2048 & 14\% & 14\% & 14\% & 14\% & 14\% & \textbf{0\%} \\
all & \textbf{0\%} & \textbf{0\%} & \textbf{0\%} & \textbf{0\%} & \textbf{0\%} & \textbf{0\%} \\
\bottomrule\end{tabular}}\end{center}

(Production in stream order at a $4\times$ cache sits at $+63\%$ for reference.)

The table is L-shaped: the excess is essentially $\min$ of a function of $W$ and a function of the capacity. \textbf{Cache capacity and reorder window are substitutes --- you need one of them, not both.} A stack with a large prefix cache can ignore ordering entirely; a stack with a one-request cache and a deep reorder queue gets the identical number. Paying for both buys nothing.

\textbf{Relation to existing schedulers.} SGLang already ships cache-aware scheduling with two policies: LPM, which prioritises the waiting request with the longest prefix match, and \texttt{DFS\_WEIGHT}, which groups waiting requests by their position in the cache tree \cite{SGL}. \texttt{DFS\_WEIGHT} \emph{is} a heuristic DFS traversal. What Theorem 7 adds is the exact statement of what that heuristic is approximating and what it is worth: the target is the DFS order of the trie, and hitting it makes the required cache capacity collapse to a single request's context --- not merely better, but exactly the unbounded-cache cost, for any layout. The practical consequence is the substitution above, which we have not seen stated: a deployment does not need to buy both a large prefix cache and a deep reorder queue.

The two ideas compose. Our hierarchy supplies the layout, its DFS traversal supplies the schedule, and both come out of the same tree --- so a scheduler already doing \texttt{DFS\_WEIGHT} is running half of this paper, against a layout chosen by something else.

\section{Related Work}

\textbf{Prefix caching in serving systems.} vLLM's PagedAttention \cite{KWO23} made KV cache a paged resource and automatic prefix caching a standard feature; SGLang's RadixAttention \cite{SGL} keeps prompt and generation KV in a radix tree with LRU eviction and adds cache-aware scheduling --- LPM, and the \texttt{DFS\_WEIGHT} policy discussed in Section 8.1. Both take the prompt as given. The work closest to ours in motivation is CacheWeaver \cite{CW26}, which states the problem exactly --- "set overlap does not become reusable prefix overlap" --- and resolves it with a greedy walk over a prefix tree, explicitly declining a combinatorial search. A second line attacks the same loss by weakening the exact-prefix rule instead of by ordering: Prompt Cache \cite{GIM24} marks up position-independent segments, CacheBlend \cite{CB24} fuses KV of non-prefix chunks with selective recomputation, CacheClip \cite{CC25} restores inter-chunk dependencies with an auxiliary model, and Cache-Craft \cite{CCR25} and PCR \cite{PCR26} manage per-chunk caches for RAG. None of these states a hardness result, an approximation guarantee, or an adaptive-versus-global separation, and none optimises a global objective over a request population.

\textbf{Sort orders in query optimisation.} The closest prior \emph{formulation} we know is Guravannavar, Sudarshan, Diwan and Sobhan Babu \cite{GSD06}, who choose sort orders for operators so that orders match across the edges of a query plan, prove a special case NP-hard, and give a 2-approximation for a join tree. The difference is precisely the content of Theorem 1: there the sharing structure --- the plan tree --- is given and only the orders are chosen; here the sharing structure \emph{is} the decision variable and the orders fall out of it. Factorised databases \cite{OZ15} make the same ordered-versus-free distinction under the name of variable orders and d-trees, where a branch-dependent order can be much smaller than any global one. Trie-size minimisation under a single global attribute permutation is classical and NP-hard, including for double-array tries \cite{KAN24}; per-record freedom is what the RAG setting hands you and what makes MPT a different object.

\textbf{Hierarchical clustering.} In savings form the objective is a sum, over internal nodes of a binary hierarchy, of a monotone-decreasing \emph{set} function of that node's leaves. This puts it outside the Dasgupta family \cite{DAS16} and its relatives \cite{CC17,MW17,CKM19}, which sum \emph{pairwise} similarities weighted by a function of cluster cardinalities --- indeed outside the admissibility characterisations for such objectives, which require the scaling to depend on the children's cardinalities. We would like to know whether the sparsest-cut machinery developed there transfers; Theorem 6's proof does not use it, and the $\tfrac12$ we get is tight for a different reason (the objective degenerates to maximum matching).

\textbf{Decision diagrams.} $\mathrm{OPT}_{\mathrm{glob}}$ versus $\mathrm{OPT}_{\mathrm{free}}$ is the ordered-versus-free distinction of branching programs \cite{BRY86,WEG00}: an OBDD fixes one variable order, an FBDD may use a different order on each path, and FBDDs can be exponentially smaller. Minimising either is hard \cite{SIE02}. Our Theorem 5 is an unbounded separation of the same shape in a setting where the tree is a clustering rather than a branching program; Open Problem 1 asks whether the known exponential separations transfer, which would settle the rate at $\Theta(k)$.

\textbf{Segment trees and consecutive ones.} $t_x(H)$ is the canonical-decomposition size that makes a range query touch $O(\log n)$ segment-tree nodes. MPT asks for the tree minimising total decomposition size over a \emph{given family} of sets; we have not found that problem studied. When the family has the consecutive-ones property a suitable leaf order is found by PQ-trees \cite{BL76}, giving the (loose) bound of Section 3.1.

\textbf{Evaluation.} Retrieval traces come from BEIR \cite{THA21} --- NFCorpus \cite{BOT16}, SciFact \cite{WAD20} and FiQA \cite{MAI18} --- with BM25 retrieval and \texttt{cl100k\_base} token counts. The quality concern that motivates Section 7.1 is the position sensitivity documented by Liu et al. \cite{LIU24}.

\section{Conclusion}

Prefix caching rewards requests that agree on a token prefix, but a growing share of production traffic is assembled from sets of reusable fragments whose order is chosen by the serving stack rather than by the user. This paper treats that choice as an optimisation problem and characterises it exactly. The minimum prefill cost is the minimum, over binary hierarchies on the request population, of the total canonical-decomposition size of the sets of requests needing each chunk; the per-request permutations are determined by that hierarchy and are not independent degrees of freedom. The characterisation supplies an exact algorithm, places the two-chunk case in correspondence with minimum vertex cover, and identifies the extremal behaviour as minimum external path length, so the cost of insisting on one global chunk order grows without bound.

The practical consequences follow from the same object. Agglomerative clustering by common intersection is a tight one-half approximation for the achievable saving. On sampled NFCorpus groups of at most 13 requests its cost is 0.05\% above the exact optimum on average; this local comparison does not determine the full-trace optimality gap. On the full retrieval traces it reduces prefill by 17--36\% over the ordering production RAG systems currently use, with 13--23\% surviving a constraint that pins the most relevant passages in place. The hierarchy also serves as a schedule, and traversing it depth-first reduces the cache capacity needed to reach the unbounded-cache optimum to a single request's context, which makes cache capacity and reorder window interchangeable resources.

Two limitations bound these claims. The evaluation measures prefill tokens rather than end-to-end latency on a serving stack, so the translation to time-to-first-token depends on hardware and batch composition that we do not model. And the quality question is priced rather than measured: we report what the saving costs under a constraint that preserves relevance order, but we do not evaluate answer quality directly.

\subsection{Open Problems}

Theorem 1 restates each of the following as a question about hierarchies rather than about permutations, which is what makes them tractable to attack.

\begin{enumerate}
\item \textbf{How much does a fixed priority cost?} $\mathrm{OPT}_{\mathrm{glob}}$ is the same minimisation restricted to hierarchies induced by a fixed chunk priority. We have $\mathrm{OPT}_{\mathrm{glob}}/\mathrm{OPT}_{\mathrm{free}} = \Omega(k/\log k)$ from $L_n$ and a trivial $O(k)$ upper bound. Search over small universes says $L_n$ is \textbf{not} extremal: the best ratios we have found by annealed search are $1.20$ at $k=3$ and $\mathbf{1.50}$ at $k=4$ ($\{0145,0236,1236,1245,1345,1456,2346,2356\}$ on 7 chunks: $14$ versus $21$), against $L_5$'s $1.167$ and $L_6$'s $1.25$ at comparable sizes. At $k=4$ the trivial upper bound is $4$ and $k/\log_2 k = 2$, so $1.50$ does not yet separate the two candidate rates, but it does say the leave-one-out family understates the gap by a wide margin. Is the truth $\Theta(k)$ or $\Theta(k/\log k)$? This is the ordered-versus-free decision diagram question in a setting where the tree is a clustering rather than a branching program, and we do not know whether the known exponential OBDD/FBDD separations transfer.
\item \textbf{Approximation of the cost.} Theorem 6 settles the savings side: greedy is a tight $\tfrac12$-approximation for $\max_H \sum_v w(I(v))$. The cost side does not follow, because $\sum_i w(R_i)/\mathrm{OPT}$ is unbounded, and we do not know whether MPT itself admits a constant factor --- or whether it is APX-hard. Note the $k=2$ case is \emph{easy} to approximate (the cost is $\tau(G)+m \in [m,2m]$), so hardness of approximation, if it holds, must come from larger requests. The related sort-order problem on a \emph{given} join tree admits a 2-approximation \cite{GSD06}, which we read as mild encouragement.
\item \textbf{Online with lookahead.} Fix a reorder window $W$. What is the competitive ratio against the offline optimum as a function of $W$? Section 7.2 measures a curve that rises monotonically and is still climbing at $W = 512$, and Section 8.1 shows $W$ trading off against cache capacity along an L-shaped frontier. We have no theory for either shape. Our guess is that $\Theta(\log W)$ of the $\log$ factor in Theorem 5 is recoverable and no more.
\item \textbf{Beyond exact prefixes.} CacheBlend \cite{CB24} and CacheClip \cite{CC25} trade accuracy for position-independent reuse. There the object is not a trie and Theorem 1 does not apply. Is there a formulation in which the accuracy loss and the layout cost are traded off in one optimisation?
\end{enumerate}

\textbf{Acknowledgements.} We thank Halil Talha Ertan for identifying the hub-cap error in the v1 spanning-tree lower-bound implementation.

\section{Appendix A. Reproducibility}

All code is at <https://github.com/Darrenus/prefix-sharing-is-sorting>. \texttt{get\_data.sh} downloads the BEIR corpora; \texttt{tiktoken} supplies the token counts.

\textbf{Version 2 certificate audit.} We reconstructed the nine reported settings using the released BM25 retrieval code and the publicly available BEIR archives on October 25, 2026: all NFCorpus and SciFact queries, the first 2\,500 FiQA queries, and $k=5,8,16$. The v1 repository did not archive its original retrieval traces; these are fresh reconstructions. The reconstructed traces, token weights, hashes, certificate checks and regression tests accompany v2 as ancillary files. The maximum observed number of requests sharing a chunk was:

\begin{center}\small
\begin{tabular}{lrrr}\toprule
corpus & $k=5$ & $k=8$ & $k=16$ \\ \midrule
NFCorpus & 40 & 51 & 71 \\
SciFact & 12 & 14 & 21 \\
FiQA & 15 & 19 & 33 \\
\bottomrule\end{tabular}\end{center}

No reconstructed setting triggers the old $>400$ cap. The released capped implementation and an independent complete-graph Prim calculation give identical spanning-tree weights in all nine settings. All nine published lower-bound fractions are therefore unchanged; the production, popularity and agglomerative cost fractions and the savings against production also reproduce to the printed precision (45 compared values in total). The corrected NFCorpus descriptive statistics in Section 7 use actual chunk-slots. These checks preserve the reported table values while fixing the implementation for workloads with more frequent shared chunks.

{\footnotesize\begin{verbatim}
python3 code/t2.py        # minimal counterexample search (Theorem 4)
python3 code/exact2.py    # exact solver vs merge-sort formula (Theorem 5)
python3 code/verify.py    # closed form to n=4000, separation table
python3 code/sim.py       # algorithm comparison, adversarial + synthetic RAG
python3 code/online2.py   # batch-level online variant
python3 code/worst.py     # worst-ratio search (open problem 1)
python3 code/engine.py    # THE STRUCTURE THEOREM: O(3^m) exact solver, L_n check
python3 code/lam.py       # laminar relaxation, cross-checked per-chunk
python3 code/tight.py     # LAM == OPT_free on 1500 random instances
python3 code/tight2.py    # ... weighted, structured, and k=2 sanity checks
python3 code/eval2.py     # greedy vs the TRUE optimum on RAG-shaped instances
python3 code/hill.py      # annealed search for the worst global/free ratio
python3 code/approx.py    # spanning-tree upper bound; greedy shortfall search
python3 code/anneal_greedy.py  # annealed search for greedy's worst case (never past 1.83)
python3 code/halfproof.py # verifies Theorem 6 AND its per-threshold engine A(t) <= 2G(t)
python3 code/phase2.py    # real BM25 retrieval traces, real token counts
python3 code/sweep.py     # 3 corpora x 3 depths; greedy vs exact on real request groups
python3 code/test_mst_bound.py   # complete-graph bound and 400/401 regression
python3 code/verify_v2_traces.py # check the nine retained v2 reconstruction traces
python3 code/tiered.py    # quality-preserving variant (pin the most relevant at the end)
python3 code/online_real.py  # lookahead sweep on the real stream
python3 code/phase3.py    # finite radix cache, real arrival order
python3 code/phase3b.py   # Theorem 7: DFS order + one-request cache
python3 code/phase3d.py   # the capacity / reorder-window substitution frontier
\end{verbatim}}

\end{document}